\documentclass[11pt]{article}

\RequirePackage[T1]{fontenc}
\RequirePackage{amssymb,amsmath,amsthm}

\theoremstyle{plain}
\newtheorem{theorem}{Theorem}
\newtheorem{lemma}[theorem]{Lemma}

\theoremstyle{definition}
\newtheorem{definition}[theorem]{Definition}

\theoremstyle{remark}

\title{Optimal Identity Testing with High Probability}
\usepackage{comment}

\newtheorem{fact}[theorem]{Fact}

\def\colorful{0}

\renewcommand{\epsilon}{\varepsilon}

\def\nnewcolor{0}
\ifnum\nnewcolor=1
\newcommand{\nnew}[1]{{\color{red} #1}}
\fi
\ifnum\nnewcolor=0
\newcommand{\nnew}[1]{#1}
\fi

\newcommand{\new}[1]{{#1}}

\ifnum\colorful=1
\newcommand{\new}[1]{{#1}}
\else
\fi

\newcommand{\R}{\mathbb{R}}
\newcommand{\Z}{\mathbb{Z}}
\newcommand{\E}{\mathbb{E}}

\newcommand{\poly}{\mathrm{poly}}

\newcommand{\dtv}{d_{\mathrm TV}}

\newcommand{\wh}[1]{{\widehat{#1}}}

\newcommand{\var}{\mathrm{Var}}

\newcommand{\ignore}[1]{}

\newcommand{\eps}{\epsilon}

\newcommand{\abs}[1]{\lvert#1\rvert}

\newcommand{\Var}{\mathop{\textnormal{Var}}\nolimits}

\newcommand{\Poi}{\mathop{\textnormal{Poi}}\nolimits}

\newcommand{\eqdef}{\stackrel{\text{\footnotesize$\mathrm{def}$}}{=}}

\usepackage{hyperref}

\usepackage[top=1in, bottom=1in, left=1in, right=1in]{geometry}

\author{
Ilias Diakonikolas\thanks{Supported by NSF Award CCF-1652862 (CAREER) and a Sloan Research Fellowship.}\\
CS, USC\\
{\tt diakonik@usc.edu}\\
\and
Themis Gouleakis\thanks{Supported by NSF Awards, including No. CCF-1650733, CCF-1733808, and IIS-1741137}\\
CSAIL, MIT\\
{\tt tgoule@mit.edu}\\
\and
John Peebles\thanks{Supported by the NSF Graduate Research Fellowship under Grant No. 1122374, and by the NSF under Grant No. 1065125.}\\ 
CSAIL, MIT\\
{\tt jpeebles@mit.edu}\\
\and
Eric Price\\
CS, UT Austin\\
{\tt ecprice@cs.utexas.edu }\\
}

\begin{document}

\maketitle

\setcounter{page}{0}

\thispagestyle{empty}

\begin{abstract}
We study the problem of testing identity against a given distribution with a focus on the high confidence regime. More precisely, given samples from an unknown distribution $p$ over $n$ elements, an explicitly given distribution $q$, and parameters $0< \epsilon, \delta < 1$, we wish to distinguish, {\em with probability at least $1-\delta$}, whether the distributions are identical versus $\eps$-far in total variation distance. Most prior work focused on the case that $\delta = \Omega(1)$, for which the sample complexity of identity testing is known to be $\Theta(\sqrt{n}/\epsilon^2)$. Given such an algorithm, one can achieve arbitrarily small values  of $\delta$ via black-box amplification,  which multiplies the required number of samples by $\Theta(\log(1/\delta))$. 

We show that \new{black-box amplification} is suboptimal for any $\delta = o(1)$, and give a new identity tester that achieves the optimal sample complexity.  Our new upper and lower bounds show that the optimal sample complexity of identity testing is
\[
  \Theta\left( \frac{1}{\epsilon^2}\left(\sqrt{n \log(1/\delta)} + \log(1/\delta) \right)\right)
\]
for any $n, \eps$, and $\delta$. For the special case of uniformity testing, where the given distribution is the uniform distribution $U_n$ over the domain,  our new tester is surprisingly simple: to test whether $p = U_n$ versus $\dtv (p, U_n) \geq \eps$, we simply threshold $\dtv(\wh{p}, U_n)$, where $\wh{p}$ is the empirical probability distribution.  The fact that this simple ``plug-in'' estimator is sample-optimal is surprising,  even in the constant $\delta$ case. Indeed, it was believed that such a tester would not attain sublinear sample complexity even for constant values of $\eps$ and $\delta$.

An important contribution of this work lies in the analysis techniques that we introduce in this context. First, we  exploit an underlying strong convexity property to bound from below  the expectation gap in the completeness and soundness cases.  Second, we give a new, fast method for obtaining provably correct  empirical estimates of the true worst-case failure probability  for a broad class of  uniformity testing statistics over all possible input distributions---including all previously studied statistics for this problem. We believe that our novel analysis techniques will be useful for other distribution testing problems as well.

\end{abstract}

\newpage

\section{Introduction}  \label{sec:intro}

Distribution property testing~\cite{GR00, BFR+:00, Batu13}, originating 
in statistical hypothesis testing~\cite{NeymanP, lehmann2005testing}, 
studies problems of the form: 
given sample access to one or more unknown distributions,
determine whether they satisfy some global property or are ``far''
from satisfying the property. (See Section \ref{ssec:defs} for a formal definition.) 
During the past two decades problems of this form have received significant attention
within the computer science community. See~\cite{Rub12, Canonne15} for two recent surveys.

Research in this field has primarily centered on determining 
tight bounds on the sample complexity of testing various properties 
\emph{in the constant probability of success regime}. That is, 
the testing algorithm must succeed with a probability of (say) at least $2/3$.
This constant confidence regime is fairly well understood. 
For a range of fundamental properties~\cite{Paninski:08, CDVV14, VV14, DKN:15, DKN:15:FOCS, ADK15, DK16, DiakonikolasGPP16}
we now have \emph{sample-optimal} testers that use provably optimal number of samples (up to constant factors) in this regime.%

In sharp contrast, the high confidence regime---i.e., the case where the desired failure probability is subconstant---%
is poorly understood even for the most basic properties.
For essentially all distribution property testing problems studied in the literature, the standard amplification method 
is the only way known to achieve a high confidence success probability.
Amplification is a black-box method that can boost the success probability to any desired accuracy. However, using it increases the number of required samples beyond what is necessary to obtain constant confidence. Specifically, to achieve a high confidence success probability of $1-\delta$ via amplification, the number of samples required increases by a factor of $\Theta(\log(1/\delta))$ compared to the constant confidence regime.

This discussion raises the following natural questions:
{\em For a given distribution property testing problem, does black-box amplification give sample-optimal testers for obtaining a high confidence success probability? Specifically, is the $\Theta (\log(1/\delta))$ multiplicative increase
in the sample size the best possible? If not, can we design testers that have optimal sample complexity in terms of all relevant problem parameters, including the error probability $\delta$?
}

We believe that these are fundamental questions that merit theoretical
investigation in their own right. As Goldreich
notes~\cite{Gold17-whp}, ``eliminating the error probability as a
parameter does not allow to ask whether or not one may improve over
the straightforward error reduction''.  From a practical perspective,
understanding this high confidence regime is important to applications
of hypothesis testing (e.g., in biology), because the failure
probability $\delta$ of the test can be reported as a $p$-value. (The
family of distribution testing algorithms with success probability
$1-\delta$ for a given problem is equivalent to the family of
statistical tests whose $p$-value (probability of Type I error) and
probability of Type II error are both at most $\delta$.)  Standard
techniques for addressing the problem of multiple comparisons, such as
Bonferroni correction, require vanishingly small $p$-values.

Perhaps surprisingly, with one exception~\cite{HM13-it}, 
this basic problem has not been previously
investigated {in the finite sample regime}. A conceptual contribution
of this work is to raise this problem as a fundamental goal in
distribution property testing.  We note here that the analogous
question in the context of {\em distribution learning} has been
intensely studied in statistics and probability theory (see,
e.g.,~\cite{VWellner96, DL:01}) and tight bounds are known in a range
of settings. 

\subsection{Formal Framework} \label{ssec:defs}

\noindent The focus of this work is on the task of
identity testing, which is arguably \emph{the} most fundamental
distribution testing problem.
 
\begin{definition}[Distribution Identity Testing Problem]
  Given a target distribution $q$ with domain {$D$} of size $n$,
  parameters $0< \eps, \delta < 1$, and sample access to an unknown
  distribution $p$ over the same domain, we want to distinguish {\em
    with probability at least $1-\delta$} between the following cases:
\begin{itemize}
\item Completeness: $p = q$.
\item Soundness: $\dtv(p, q) \geq \eps$.
\end{itemize}
We call this the problem of {\em $(\eps, \delta)$ testing identity to
  $q$}.  The special case of $q$ being uniform is known as
\emph{uniformity testing}.  An algorithm that solves one of these
problems will be called an \emph{$(\varepsilon,\delta)$-tester} for
identity/uniformity.
\end{definition}

\noindent Note that $\dtv(p, q)$ denotes the total variation distance
or statistical distance between distributions $p$ and $q$, i.e.,
$\dtv(p, q) \eqdef \frac{1}{2} \cdot \|p-q\|_1$.  The goal is to
characterize the sample complexity of the problem: i.e., the number of
samples that are necessary and sufficient to correctly distinguish
between the completeness and soundness cases with  probability
$1-\delta$.

\subsection{Our Results} \label{ssec:results}

Our main result is a complete characterization of the worst-case sample complexity 
of identity testing in the high confidence regime. For this problem, 
we show that black-box amplification is suboptimal for {\em any} $\delta = o(1)$, 
and give a new identity tester that achieves the optimal sample complexity:

\begin{theorem}[Main Result] \label{thm:main}
There exists a computationally efficient $(\eps, \delta)$-identity tester for discrete distributions of support size $n$ with sample complexity 
\begin{align}\label{eq:samplecomplexity}
\Theta\left(\frac{1}{\eps^2}\left(\sqrt{n \log(1/\delta)} + \log(1/\delta)\right)\right).
\end{align}
Moreover, this sample size is information-theoretically optimal, up to a constant factor, for all 
$n, \eps, \delta$.
\end{theorem}

\new{ As we explain in Section~\ref{ssec:discussion},~\cite{HM13-it}
  gave a tester that achieves the optimal sample complexity when the
  sample size is $o(n)$.  However this tester {\em completely} fails
  with $\Omega(n)$ samples, as may be required when either $\eps$ or
  $\delta$ are sufficiently small.  Theorem~\ref{thm:main} provides a complete
  characterization of the worst-case sample complexity of the problem
  with a {\em single} statistic for all settings of parameters
  $n, \eps, \delta$.

\smallskip

\noindent {\bf Brief Overview of Techniques.} 
To analyze our tester, we introduce two new techniques for the
analysis of distribution testing statistics, which we describe in more
detail in Section \ref{ssec:techniques}. Our techniques leverage a
simple common property of numerous distribution testing statistics
which does not seem to have been previously exploited in their
analysis: their convexity.  Our first technique crucially exploits an
underlying strong convexity property to bound from below the
expectation gap between the completeness and soundness cases. We remark
that this is a contrast to most known distribution testers where
bounding the expectation gap is easy, and the challenge is in bounding
the variance of the statistic.

Our second technique implies a new, fast method for obtaining empirical estimates 
of the true worst-case failure probability of any member of a broad class of uniformity testing statistics. 
This class includes all uniformity testing statistics studied in the literature. Critically, these estimates 
come with provable guarantees about the worst-case failure probability of the statistic 
over all possible input distributions, and have tunable additive error. We elaborate in Section \ref{ssec:techniques}.
}

\subsection{Discussion and Prior Work} \label{ssec:discussion}
Uniformity testing is the first and one of the most well-studied problems in
distribution testing~\cite{GR00, Paninski:08, VV14, DKN:15, DiakonikolasGPP16}. 
As already mentioned, the literature has almost exclusively focused on the case of constant error probability 
$\delta$.  The first uniformity tester, introduced by Goldreich and Ron~\cite{GR00}, 
counts the number of collisions among the samples and was shown to work with
$O(\sqrt{n}/\eps^4)$ samples~\cite{GR00}.  A related tester proposed by Paninski~\cite{Paninski:08}, 
which relies on the number of distinct elements in the set of samples, was shown to have
the optimal $m = \Theta(\sqrt{n}/\eps^2)$ sample complexity, as long as $m = o(n)$.  
Recently, a chi-squared based tester was shown in~\cite{VV14, DKN:15} to
achieve the optimal $\Theta(\sqrt{n}/\eps^2)$ sample complexity without any
restrictions.  Finally, the original collision-based tester of~\cite{GR00} was very recently shown
to also achieve the optimal $\Theta(\sqrt{n}/\eps^2)$ sample
complexity~\cite{DiakonikolasGPP16}.  Thus, the situation for constant
values of $\delta$ is well understood.

{The problem of identity testing against an arbitrary (explicitly given) distribution was studied in~\cite{BFFKRW:01}, who gave an $(\eps, 1/3)$-tester
with sample complexity $\tilde{O}(n^{1/2})/\poly(\eps)$. The tight bound of $\Theta(n^{1/2}/\eps^2)$ was first given in~\cite{VV14} 
using a chi-squared type tester (inspired by~\cite{CDVV14}). In subsequent work, a similar chi-squared tester that also achieves 
the same sample complexity bound was given in~\cite{ADK15}.
(We note that the~\cite{VV14, ADK15} testers have sub-optimal sample complexity in the high confidence regime, even for the case of uniformity.)
In a related work, \cite{DKN:15} obtained a reduction of identity to uniformity 
that preserves the sample complexity, up to a constant factor, in the constant error probability regime.
More recently, Goldreich~\cite{Goldreich16}, building on~\cite{DK16}, gave a different reduction of identity to uniformity that preserves the error probability.
We use the latter reduction in this paper to obtain an optimal identity tester starting from our new optimal uniformity tester.
}

Since the sample complexity of identity testing is $\Theta(\sqrt{n}/\eps^2)$ for $\delta = 1/3$~\cite{VV14, DKN:15},
standard amplification gives a sample upper bound of $\Theta(\sqrt{n} \log(1/\delta)/\eps^2)$ for this problem.
It is not hard to observe that this naive bound cannot be optimal for all values of $\delta$. For example, in the extreme
case that $\delta = 2^{-\Theta(n)}$, this gives a sample complexity of $\Theta(n^{3/2}/\eps^2)$.
On the other hand, one can {\em learn} the underlying distribution 
(and therefore test for identity) with $O(n/\eps^2)$ samples for such values of $\delta$\footnote{This follows from the fact that, 
for any distribution $p$ over $n$ elements, the empirical probability distribution $\widehat{p}_m$ obtained after 
$m = \Omega((n+\log(1/\delta))/\eps^2)$ samples drawn from $p$ 
is $\eps$-close to $p$ in total variation distance with probability at least $1-\delta$.}.

{The case where $1 \gg \delta \gg 2^{-\Theta(n)}$ is more subtle,
  and it is not a priori clear how to improve upon naive
  amplification.  Theorem~\ref{thm:main} provides a smooth transition
  between the extremes of $\Theta(\sqrt{n}/\eps^2)$ for constant
$\delta$ and $\Theta(n/\eps^2)$ for $\delta = 2^{-\Theta(n)}$.  It
thus provides a quadratic improvement in the dependence on $\delta$
over the naive bound for all $\delta \geq 2^{-\Theta(n)}$, and shows
that this is the best possible.  For $\delta < 2^{-\Theta(n)}$, it
turns out that the additive $\Theta(\log(1/\delta)/\eps^2)$ term is
necessary, as outlined in Section~\ref{ssec:techniques}, so learning
the distribution is optimal for such values of $\delta$.}

We obtain the first sample-optimal
{\em uniformity tester} for the high confidence regime. Our
sample-optimal identity tester follows from our uniformity tester by applying the recent result of
Goldreich~\cite{Goldreich16}, which provides a black-box reduction of identity to uniformity. 
We also show a matching information-theoretic lower bound on the sample complexity.

The sample-optimal uniformity tester we introduce is remarkably
simple: to distinguish between the cases that $p$ is the uniform distribution $U_n$ over $n$ elements 
versus $\dtv(p, U_n) \geq \eps$, we simply compute
$\dtv(\wh{p}, U_n)$ for the empirical distribution $\wh{p}$.
The tester accepts that $p = U_n$ if the value of this statistic is below some
well-chosen threshold, and rejects otherwise.

It should be noted that such a tester 
was not previously known to work with sub-learning sample complexity---i.e., fewer than $\Theta(n/\epsilon^2)$ samples---%
even in the constant confidence regime. Surprisingly, in a literature with several different
uniformity testers~\cite{GR00, Paninski:08, VV14, DKN:15}, no one has previously used the empirical total variation distance. {On the contrary, it would be natural 
to assume---as was suggested in~\cite{BFR+:00, Batu13}---that this tester cannot possibly work.} 
A likely reason for this is the following observation: When
the sample size $m$ is smaller than the domain size  $n$, the empirical total variation 
distance is very far from the true distance to uniformity. 
This suggests that the empirical distance statistic gives little, if any, information in this setting.

Despite the above intuition, we prove that the natural ``plug-in'' estimator
relying on the empirical distance from uniformity actually works for the following reason:
the empirical distance from uniformity is {\em noticeably smaller} for the uniform distribution 
than for ``far from uniform'' distributions, {\em even with a sub-linear sample size}. 
Moreover, we obtain the stronger statement that the ``plug-in'' estimator is 
a sample-optimal uniformity tester for all parameters $n$, $\eps$ and $\delta$.

In~\cite{HM13-it}, it was shown that the distinct-elements tester
of~\cite{Paninski:08} achieves the optimal sample complexity of
$m = \Theta(\sqrt{n \log(1/\delta)}/\eps^2)$ , as long as $m = o(n)$.
When $m = \Omega(n)$, as is the case in many practically relevant
settings (see, e.g., the Polish lottery example in \cite{Rubinfeld14}
with $n< \sqrt{n}/\eps^2 \new{\ll n/\eps^2}$), this tester is known to
fail completely even in the constant confidence regime.  \new{On the
  other hand, in such settings the sample size is {\em not}
  sufficiently large so that we can actually {\em learn} the
  underlying distribution.}

It is important to note that {\em all} previously considered
uniformity testers~\cite{GR00, Paninski:08, VV14, DKN:15} do not
achieve the optimal sample complexity (as a function of all
parameters, including $\delta$), and this is {\em inherent}, i.e., not
just a failure of previous analyses.  Roughly speaking, since the
collision statistic~\cite{GR00} and the chi-squared based
statistic~\cite{VV14, DKN:15} are not Lipschitz, it can be shown that
their high-probability performance is poor.  Specifically, in the
completeness case ($p = U_n$), if many samples happen to land in the
same bucket (domain element), these test statistics become quite
large, leading to their suboptimal behavior {\em for all}
$\delta = o(1)$. {(For a formal justification, the reader is referred
  to Section V of~\cite{HM13-it}).}  On the other hand, the
distinct-elements tester~\cite{Paninski:08} does not work for
$m = \omega(n)$.  For example, if $\eps$ or $\delta$ are sufficiently
small to necessitate $m \gg n \log n$, then typically all $n$ domain
elements will appear in both the completeness and soundness cases,
hence the test statistic provides no information.

\subsection{Our Techniques} \label{ssec:techniques}

\subsubsection{Upper Bound for Uniformity Testing} We would like to show
that the test statistic $\dtv(\wh{p}, U_n)$ is with high
probability larger when $\dtv(p, U_n) \geq \eps$ than when
$p = U_n$.  We start by showing that among all possible
alternative {distributions $p$  with $\dtv(p, U_n) \geq \eps$}, it suffices to consider those in a very
simple family.  We then show that the test statistic is highly
concentrated around its expectation, and that the expectations are
significantly different in the two cases. The main technical components of our paper are our techniques for accomplishing these tasks.

To simplify the structure of $p$, we show in {(Section~\ref{sec:sd})}
that if $p$ majorizes another distribution $q$, then the test
statistic $\dtv(\wh{p}, U_n)$ stochastically dominates
$\dtv(\wh{q}, U_n)$.  (In fact, this statement holds for {\em any}
test statistic that is a convex symmetric function of the empirical
histogram.)  Therefore, for any $p$, if we average out the large and
small entries of $p$, the test statistic becomes harder to distinguish
from uniform.

We remark as a matter of independent interest that this stochastic
domination lemma immediately implies a fast algorithm for performing
rigorous empirical comparisons of test statistics.  A major difficulty
in empirical studies of distribution testing is that it is not
possible to directly check the failure probability of a tester over
every possible distribution as input, because the space of such
distributions is quite large. Our {structural} lemma reduces the
search space dramatically for uniformity testing: for any convex
symmetric test statistic (which includes all existing ones), the worst
case distribution will have $\alpha n$ coordinates of value
$(1 + \eps/\alpha)/n$, and the rest of value
$(1 - \eps/(1-\alpha))/n$, for some $\alpha$.  Hence, there are only
$n$ possible worst-case distributions for any $\eps$. Notably, this
reduction does not lose anything, so it could be used to identify the
non-asymptotic optimal constants that a distribution testing statistic
achieves for a given set of parameters.

Returning to our uniformity tester, at the cost of a constant factor
in $\eps$ we can assume $\alpha = 1/2$.  As a result, we only need to
consider $p$ to be either $U_n$ or of the form $\frac{1 \pm \eps}{n}$
in each coordinate.  We now need to separate the expectation of the
test statistic in these two situations.  The challenge is that both
expectations are large, and we do not have a good analytic handle on
them.  We therefore introduce a new technique for showing a separation
between the completeness and soundness cases that utilizes the strong
convexity of the test statistic.  Specifically, we obtain an explicit
expression for the Hessian of the expectation, as a function of
$p$. The Hessian is diagonal, and for our two situations of
$p_i \approx 1/n$ each entry is within constant factors of the same
value, giving a lower bound on its eigenvalues.  Since the expectation
is minimized at $p = U_n$, strong convexity implies an expectation
gap. Specifically, we prove that this gap is
$\eps^2 \cdot \min(m^2/n^2, \sqrt{m/n}, 1/\eps)$.

Finally, we need to show that the test statistic concentrates about
its expectation.  For $m \geq n$, this follows from McDiarmid's inequality:
since the test statistic is $1/m$-Lipschitz in the $m$ samples, with
probability $1 - \delta$ it lies within $\sqrt{\log(1/\delta) / m}$ of
its expectation. When $m$ is larger than the desired sample
complexity given in~\eqref{eq:samplecomplexity}, this is less than the
expectation gap above. The concentration is trickier when $m < n$,
since the expectation gap is smaller, so we need to establish tighter concentration.
We get this by using a Bernstein variant of McDiarmid's inequality, {which is stronger
than the standard version of McDiarmid in this context.} We note that the use of the stochastic domination is also crucial here. Since our statistic is a symmetric convex function of the histogram values, we can use lemma \ref{lm:general domination} to assume without loss of generality that the soundness case distribution has possible probability mass values exclusively in the set $\{\frac{1+\varepsilon^\prime}{n},\frac1n,\frac{1-\varepsilon^\prime}{n}\}$, for some $\varepsilon^\prime = O(\varepsilon)$. This distribution has a stronger Lipschitz-type property than the other soundness case distributions. Therefore, we are able to use a stronger concentration bound via McDiarmid's inequality and argue that even though other soundness case distributions may have weaker concentration, they still have smaller error due to our stochastic domination argument.

\subsubsection{Upper Bound for Identity Testing} In~\cite{Goldreich16}, it was
shown how to reduce $\eps$-testing of an arbitrary distribution $q$
over $[n]$ to $\eps/3$-testing of $U_{6n}$.  This reduction preserves
the error probability $\delta$, so applying it gives an identity tester 
with the same sample complexity as our uniformity tester, up to constant factors.

\subsubsection{Sample Complexity Lower Bound} To match our upper
bound~\eqref{eq:samplecomplexity}, we need two lower bounds.  The
lower bound of $\Omega(\frac{1}{\eps^2} \log (1/\delta))$ is
straightforward from the same lower bound as for distinguishing a fair
coin from an $\eps$-biased coin, while the
$\sqrt{n \log (1/\delta)}/\eps^2$ bound is more challenging.

For intuition, we start with a $\sqrt{n \log(1/\delta)}$ lower bound
for constant $\eps$.  When $p = U_n$, the chance that all $m$ samples
are distinct is at least $(1-m/n)^m \approx e^{-m^2/n}$.  Hence, if
$m \ll \sqrt{n \log(1/\delta)}$, this would happen with {probability}
significantly larger than $2\delta$.  On the other hand, if $p$ is
uniform over a random subset of $n/2$ coordinates, the $m$ samples
will also all be distinct with probability $(1 - 2m/n)^m > 2\delta$.
The two situations thus look the same with $2\delta$ probability, so
no tester could have accuracy $1-\delta$.

This intuition can easily be extended to include a $1/\eps$
dependence, but getting the desired $1/\eps^2$ dependence requires
more work.  First, we Poissonize the number of samples, so we
independently see $\Poi(m p_i)$ samples of each coordinate $i$; with
exponentially high probability, this Poissonization only affects the
sample complexity by constant factors.  Then, in the alternative
hypothesis, we set each $p_i$ independently at random to be
$\frac{1 \pm \eps}{n}$.  This has the unfortunate property that $p$ no
longer sums to $1$, so it is a ``pseudo-distribution'' rather than an
actual distribution.  Still, it is exponentially likely to sum to
$\Theta(1)$, and using techniques from~\cite{WY16,DK16} this is
sufficient for our purposes.

At this point, we are considering a situation where the number of
times we see each coordinate is either $\Poi(m/n)$ or
$\frac{1}{2}(\Poi((1-\eps)\frac{m}{n}) + \Poi((1+\eps)\frac{m}{n}))$,
and every coordinate is independent of the others.  These two
distributions have Hellinger distance at least $\eps^2m/n$ in each
coordinate. 
Then the composition property for Hellinger distance 
over $n$ independent coordinates implies
$m \geq \sqrt{n \log(1/\delta)}/\eps^2$ is necessary for success
probability $1-\delta$.

\subsection{Notation}

We write $[n]$ to denote the set $\{1, \ldots, n\}$.
We consider discrete distributions over $[n]$, which are functions
$p: [n] \rightarrow [0,1]$ such that $\sum_{i=1}^n p_i =1.$
We use the notation $p_i$ to denote the probability of element
$i$ in distribution $p$. {For $S \subseteq [n]$, we will denote $p(S) = \sum_{i \in S} p_i$.}
We will also sometimes think of $p$ as an $n$-dimensional vector. We will denote by $U_n$ the uniform distribution over $[n]$.

For $r \ge 1$, the $\ell_r$-norm of a distribution is identified with the $\ell_r$-norm of the corresponding vector, i.e.,
$\|p\|_r = \left(\sum_{i=1}^n |p_i|^r\right)^{1/r}$. The $\ell_r$-distance between distributions
$p$ and $q$ is defined as the  the $\ell_r$-norm of the vector of their difference. 
{The total variation distance between distributions $p$ and $q$ is defined as $\dtv(p, q) \eqdef \max_{S \subseteq [n]} |p(S) - q(S)| = (1/2) \cdot \|p-q\|_1$.
The Hellinger distance between $p$ and $q$ is 
$H(p, q) \eqdef (1/\sqrt{2}) \cdot \| \sqrt{p} - \sqrt{q} \|_2 = (1/\sqrt{2}) \cdot \sqrt{\sum_{i=1}^n (\sqrt{p_i} - \sqrt{q_i})^2}$.
We denote by $\mathrm{Poi}(\lambda)$ the Poisson distribution with parameter $\lambda$.
}

{
\subsection{Structure of this Paper} 
In Section~\ref{sec:uniform}, we formally describe and analyze our sample-optimal uniformity tester.
In Section~\ref{sec:lb}, we give our matching sample complexity lower bound.
Finally, Section~\ref{sec:sd} establishes our stochastic domination result that is crucial for the analysis
of the soundness in Section~\ref{sec:uniform}, and may be useful in the rigorous empirical evaluation
of test statistics.
}

\section{Sample-Optimal Uniformity Testing} \label{sec:uniform}
In this section, we describe and analyze our optimal uniformity tester.
Given samples from an unknown distribution $p$ over $[n]$, 
our tester returns ``YES'' with probability $1-\delta$ if $p = U_ n$, 
and ``NO'' with  probability $1-\delta$ if $\dtv(p,U_n)\geq \eps$. 

\subsection{Our Test Statistic} \label{ssec:empirical-test}
We define a very natural statistic that yields a uniformity tester with optimal 
dependence on the domain size $n$, the proximity parameter $\varepsilon$, and the error probability $\delta$. 
Our statistic is a thresholded version of the empirical total variation distance between the unknown distribution $p$ and the uniform distribution. 
Our tester \textsc{Test-Uniformity} is described in the following pseudocode:

\medskip
\fbox{\parbox{6in}{

{\bf Algorithm} \textsc{Test-Uniformity$(p, n, \eps,\delta)$} \\
Input: sample access to a distribution $p$ over $[n]$, $\eps>0$, and $\delta>0$.\\
Output: ``YES'' if $p =U_n$; ``NO'' if $\dtv(p, U_n) \geq \eps.$
\begin{enumerate}
  \item Draw $m=\Theta\left(  (1/\eps^2) \cdot \Big(\sqrt{n \log (1/\delta)} + \log(1/\delta)\Big)\right)$ i.i.d. samples from $p$.
  
  \item Let $X=(X_1,X_2,\dots , X_n) \in \Z_{>0}^n$ be the histogram of the samples. 
  That is, $X_i$ is the number of times domain element $i$ appears in the (multi-)set of samples. 
  
  \item Define the random variable $S = \frac12 \sum_{i=1}^n \left\vert \frac{X_i}{m}-\frac1n \right\vert$ 
  and set the threshold \[t=\mu(U_n)+C \cdot
   \begin{cases}
   \epsilon^2 \cdot \frac{m^2}{n^2} & \textrm{ for } m \leq n\\
   \epsilon^2 \cdot \sqrt{\frac{m}{n}} & \textrm{ for } n < m \leq \frac{n}{\epsilon^2} \\
   \epsilon & \textrm{ for } \frac{n}{\epsilon^2} \leq m\\ 
   \end{cases} \;, \] where $C$ is a universal constant (derived from the analysis of the algorithm), 
  and $\mu(U_n)$ is the expected value of the statistic in the completeness case. (We can compute $\mu(U_n)$ in $O(m)$ time using the procedure in Appendix~\ref{app:exp}.) 

 \item If $S \ge t$ return ``NO''; otherwise, return ``YES''.
\end{enumerate}
}}

\bigskip

The main part of this section is devoted to the analysis of  \textsc{Test-Uniformity}, 
establishing the following theorem:

\begin{theorem}\label{thm:main-unif}
There exists a universal constant $C>0$ such that the following holds:
Given $$m \geq C \cdot (1/\eps^2) \left(\sqrt{n \log (1/\delta)} + \log(1/\delta)\right)$$ samples from an unknown distribution $p$,
Algorithm \textsc{Test-Uniformity} is an $(\eps, \delta)$-tester for uniformity of distribution $p$.
\end{theorem}

{
As we point out in Appendix~\ref{app:exp}, the value $\mu(U_n)$ can be computed efficiently, hence our overall tester
is computationally efficient.}
To prove correctness of the above tester, we need to show that the expected value of the statistic 
in the completeness case is sufficiently separated from the expected value in the soundness case, 
and also that the value of the statistic is highly concentrated around its expectation in both cases. 
In Section~\ref{ssec:exp-gap}, we bound from below the difference in the expectation
of our statistic in the completeness and soundness cases. In Section~\ref{ssec:concen}, 
we prove the desired concentration which completes the proof of Theorem~\ref{thm:main-unif}.
  
\subsection{Bounding the Expectation Gap} \label{ssec:exp-gap}
The expectation of the statistic in algorithm \textsc{Test-Uniformity} can be viewed 
as a function of the $n$ variables $p_1,\dots ,p_n$. 
We denote this expectation by $\mu(p) \eqdef \mathbb{E}[S(X_1, \ldots, X_n)]$ 
when the samples are drawn from distribution $p$.  

{Our analysis has a number of complications for the following reason:}
the function $\mu(p)-\mu(U_n)$ is a linear combination of sums 
that have no indefinite closed form, even if the distribution $p$ assigns only two possible 
probabilities to the elements of the domain. This statement is made precise in Appendix~\ref{app:no-closed-form}. 
As such, we should only hope to obtain an approximation of this quantity.

A natural approach to try and obtain such an approximation would be to produce 
separate closed form approximations for $\mu(p)$ and $\mu(U_n)$, and combine these quantities to obtain an approximation for their difference.
However, one should not expect such an approach to work in our context.
The reason is that the difference $\mu(p)-\mu(U_n)$ can be much smaller than $\mu(p)$ and $\mu(U_n)$; it can even be arbitrarily small.
As such, obtaining separate approximations of $\mu(p)$ and $\mu(U_n)$ to any fixed accuracy would contribute too much error to their difference.

To overcome these difficulties, we introduce the following technique, which is novel in this context.
We directly bound from below the difference $\mu(p)-\mu(U_n)$ using {\em strong convexity}. 
{Specifically, we show that the function $\mu$ is strongly convex with appropriate parameters and use this fact
to bound the desired expectation gap. The main result of this section is the following lemma:}

\begin{lemma}\label{lem:expectation-gap}
Let $p$ be a distribution over $[n]$ and $\epsilon = \dtv(p, U_n)$. 
{For all $m \geq 6$ and $n \geq 2$}, we have that:
$$
\mu(p) - \mu(U_n) \geq \Theta(1) \cdot
\begin{cases}
\epsilon^2 \cdot \frac{m^2}{n^2} & \textrm{ for } m \leq n\\
\epsilon^2 \cdot \sqrt{\frac{m}{n}} & \textrm{ for } n < m \leq \frac{n}{\epsilon^2} \\ 
\epsilon & \textrm{ for } \frac{n}{\epsilon^2} \leq m
\end{cases} \;.
$$
\end{lemma}

{We note that the bounds in the right hand side above are tight, up to constant factors. Any 
asymptotic improvement would yield a uniformity tester with sample complexity that violates our tight
information-theoretic lower bounds.}

{The proof of Lemma~\ref{lem:expectation-gap} (which will be deferred to appendix \ref{appendix:exp-gap lemma}) 
requires a couple of important intermediate lemmas.}
Our starting point is as follows: By the intermediate value theorem, we have the quadratic expansion
\[
\mu(p)=\mu(U_n)+\nabla \mu(U_n)^\intercal (p-U_n ) + \frac12 (p-U_n)^\intercal H_{p'}(p-U_n) \;,
\]
where $H_{p'}$ is the Hessian matrix of the function $\mu$ at some point $p'$ which lies on the line segment between $U_n$ and $p$.
This expression can be simplified as follows: {First, we show (Fact~\ref{cor:mean-domination}) that} our $\mu$ 
is minimized over all probability distributions on input $U_n$. 
Thus, the gradient $\nabla \mu(U_n)$ must be orthogonal to being a {direction in the space of probability distributions}.
{In other words, $\nabla \mu(U_n)$ must be proportional to the all-ones vector.}
More formally, {since $\mu$ is symmetric its gradient is a symmetric function, which implies
it will be symmetric when given symmetric input.}
Moreover, $(p-U_n)$ is a direction within the space of probability distributions, 
and therefore sums to $0$, making it orthogonal to the all-ones vector. 
Thus, we have that 
$\nabla \mu(U_n)^\intercal (p-U_n )=0$, and we obtain
\begin{equation}\label{eq:sc-lb}
\mu(p) - \mu(U_n) = \frac12 (p-U_n)^\intercal H_{p'}(p-U_n) \geq \frac12 \|p-U_n\|_2^2 \cdot \sigma \geq  \frac{1}{2} \|p-U_n\|_1^2 / n \cdot \sigma \;,
\end{equation}
where $\sigma$ is the minimum eigenvalue of the Hessian of $\mu$ 
on the line segment between $U_n$ and $p$.

{The majority of this section is devoted to proving a lower bound for $\sigma$.}
Before doing so, however, we must first address a technical consideration. 
Because we are considering a function over the space of probability distributions---which is not full-dimensional---the Hessian and gradient of $\mu$ with respect to $\R^n$ depend not only on the definition of our statistic $S$, but also its parameterization. 
For the purposes of this subsection, we parameterize $S$ as \nnew{$S(x) = \sum_{i=1}^n \max \left\{ \frac{x_i}{m} - \frac{1}{n},0 \right\}= \frac{1}{m} \sum_{i=1}^n \max \left\{ x_i - \frac{m}{n},0 \right\}$.

In the analysis we are about to perform, it will be helpful to replace $\frac{m}{n}$ with a free parameter $t$ which we will eventually set back to roughly	 $m/n$. Thus, we define }
\[
S_t(x) \triangleq \frac{1}{m} \sum_{i=1}^n \max \{x_i - t,0\}
\]
and
\begin{equation}\label{eq:expectation-formula}
\mu_t(p) \triangleq \E_{x \sim \text{Multinomial}(m,p)}[S_t(x)] = \frac{1}{m}\sum_{i=1}^n \sum_{k=\lceil t \rceil}^m   \binom{m}{k}   p_i^k   (1-p_i)^{m-k}  (k-t) \;.
\end{equation}

Note that when $t=m/n$ we have $S_t = S$ and $\mu_t=\mu$. 
Also note that when we compute the Hessian of $\mu_t(p)$, we are treating $\mu_t(p)$ as a function of $p$ and not of $t$. 
In the following lemma, we derive an exact expression for the entries of the Hessian. This result is perhaps surprising 
in light of the likely nonexistence of a closed form expression for $\mu(p)$. That is, 
while the expectation $\mu(p)$ may have no closed form, we prove that the Hessian of $\mu(p)$ does in fact have a closed form.

\begin{lemma}\label{lem:hessian-exact}
The Hessian of $\mu_t(p)$ viewed as a function of $p$ is a diagonal matrix whose $i$th diagonal entry is given by
\[
h_{ii} =  s_{t,i} \;,
\]
where we define $s_{t,i}$ as follows: 
Let $\Delta t$ be the distance of $t$ from the next largest integer, i.e., $\Delta t \triangleq \lceil t \rceil - t$. 
Then, we have that
\[
s_{t,i} = 
\begin{cases} 
0 & \textrm{ for } t=0 \\ 
(m-1){m-2 \choose t-1}p_i^{t-1}(1-p_i)^{m-t-1} &  \textrm{ for }  t \in \mathbb{Z}_{>0} \\ 
\Delta t \cdot s_{\lfloor t \rfloor,i} + (1-\Delta t) \cdot s_{\lceil t \rceil,i} &  \textrm{ for } t \geq 0 \text{ and } t \not \in \Z
\end{cases} \;.
\]
\end{lemma}

In other words, we will derive the formula for integral $t \geq 1$ and then prove that the value for nonintegral $t \geq 0$ can be found by linearly interpolating between the closest integral values of $t$.

\begin{proof}
Note that because $S_t(x)$ is a separable function of $x$, $\mu_t(p)$ is a separable function of $p$, 
and hence the Hessian of $\mu_t(p)$ is a diagonal matrix. 
By Equation~\ref{eq:expectation-formula}, the $i$-th diagonal entry of this Hessian can be written explicitly as the following expression:
\[
s_{t,i} = \frac{\partial^2}{\partial p_i^2}\mu_t(p)=\frac{d^2}{dp_i^2}\frac{1}{m}\sum_{k=\lceil t\rceil}^m   \binom{m}{k}   p_i^k   (1-p_i)^{m-k}  (k-t) \;.
\]
Notice that if we sum starting from $k=0$ instead of $k=\lceil t \rceil$, then the sum equals the expectation of $\text{Bin}(m,p_i)$ minus $t$. That is, notice that:
\[
\frac{\mathrm{d}^2}{\mathrm{d}p_i^2} \frac{1}{m} \sum_{k=0}^m   \binom{m}{k}   p_i^k   (1-p_i)^{m-k}  (k-t) =  \frac{\mathrm{d}^2}{\mathrm{d}p_i^2} \frac{1}{m} (p_i m - t) = 0 \;.
\]
By this observation and the fact that the summand is $0$ for integer $t$ when $k=t$, 
we can switch which values of $k$ we are summing over to $k$ from $0$ through $\lfloor t \rfloor$ if we negate the expression:
\[
s_{t,i} = \frac{\partial^2}{\partial^2p_i}\mu_t(p)=\frac{1}{m}  \frac{\mathrm{d}^2}{\mathrm{d}p_i^2} \sum_{k=0}^{\lfloor t \rfloor}   \binom{m}{k}   p_i^k   (1-p_i)^{m-k}  (t-k) \;.
\]

We first prove the case when $t \in \mathbb{Z}_{+}$. In this case, we view $s_{t,i}$ as a sequence with respect to $t$ (where $i$ is fixed), which we denote $s_t$. We now derive a generating function for this sequence.\footnote{To avoid potential convergence issues, we view generating functions as formal polynomials from the ring of infinite formal polynomials. Under this formalism, there is no need to deal with convergence at all.} Observe that derivatives that are not with respect to the formal variable commute with taking generating functions. Then, the generating function for the sequence $\{s_t\}$ is
\[
\frac{\mathrm{d}^2}{\mathrm{d}p_i^2} \frac{1}{m} \left( x\frac{\mathrm{d}}{\mathrm{d}x} \left( \frac{(p_i x+1-p_i)^m}{1-x}\right) - \frac{x \frac{\mathrm{d}}{\mathrm{d}x} (p_i x+1-p_i)^m}{1-x} \right) = (m-1) (p_i x+1-p_i)^{m-2} x \;.
\]
Note that the coefficient on $x^0$ is $0$, so $s_{0,i}=0$ as claimed. 
For $t \in \Z_{>0}$, the right hand side is the generating function of
\[
(m-1) \binom{m-2}{t-1} p^{t-1} (1-p)^{m-t-1} \;.
\]
Thus, this expression gives the $i$-th entry Hessian in the $t \in \mathbb{Z}_{\geq 0}$, as claimed.

Now consider the case when $t$ is not an integer. In this case, we have:
\begin{align*}
s_{t,i} &\triangleq \frac{\mathrm{d}^2}{\mathrm{d}p_i^2} \frac{1}{m} \sum_{k=\lceil t \rceil}^m   \binom{m}{k}   p_i^k   (1-p_i)^{m-k}  (k-t) \\
&= \frac{\mathrm{d}^2}{\mathrm{d}p_i^2} \frac{1}{m} \sum_{k=\lceil t \rceil}^m   \binom{m}{k}   p_i^k   (1-p_i)^{m-k}  (k-\lceil t \rceil + \Delta t)\\
&= s_{\lceil t \rceil,i} + \Delta t \frac{\mathrm{d}^2}{\mathrm{d}p_i^2} \frac{1}{m} \sum_{k=\lceil t \rceil}^m   \binom{m}{k}   p_i^k   (1-p_i)^{m-k}.\\
&= s_{\lceil t \rceil,i} - \Delta t \frac{\mathrm{d}^2}{\mathrm{d}p_i^2} \frac{1}{m} \sum_{k=0}^{\lceil t \rceil-1}   \binom{m}{k}   p_i^k   (1-p_i)^{m-k} \;.
\end{align*}
The last equality is because if we change bounds on the sum so they are from $0$ through $m$, we get $1$ which has partial derivative $0$. Thus, we can flip which terms we are summing over if we negate the expression.

Note that this expression we are subtracting above can be alternatively written as:
\[
\Delta t \frac{\mathrm{d}^2}{\mathrm{d}p_i^2} \frac{1}{m} \sum_{k=0}^{\lceil t \rceil-1}   \binom{m}{k}   p_i^k   (1-p_i)^{m-k} = \Delta t \cdot (s_{\lceil t \rceil,i} - s_{\lfloor t \rfloor,i}) \;.
\]
Thus, we have
\[
s_{t,i} = s_{\lceil t \rceil,i} - \Delta t \cdot (s_{\lceil t \rceil,i} - s_{\lfloor t \rfloor,i}) = \Delta t \cdot s_{\lfloor t \rfloor,i} + (1-\Delta t) \cdot s_{\lceil t \rceil,i} \;,
\]
as desired. This completes the proof of Lemma~\ref{lem:hessian-exact}.
\end{proof}

It will be convenient to simplify the exact expressions of Lemma~\ref{lem:hessian-exact} into something more manageable.
This is done in the following lemma:
\begin{lemma}\label{lem:hessian-approx}
Fix any constant $c>0$. The Hessian of $\mu(p)$, viewed as a function of $p$, is a diagonal matrix 
whose $i$-th diagonal entry is given by
\[
h_{ii} = s_{t: = m/n, i} \geq \Theta(1) \cdot 
\begin{cases}
\frac{m^2}{n} & \textrm{ for } m \leq n\\ 
\sqrt{m n} & \textrm{ for } n < m \leq c \cdot \frac{n}{\epsilon^2}
\end{cases} \;,
\]
assuming $p_i=\frac{1 \pm \epsilon}{n}$, $m \geq 6$, $n \geq 2$, and $\epsilon \leq 1/2$.
\end{lemma}

Similarly, these bounds are tight up to constant factors, as further improvements 
would violate our sample complexity lower bounds.

\begin{proof}
By Lemma~\ref{lem:hessian-exact}, we have an exact expression $s_{t,i}$ 
for the $i$th entry of the Hessian of $\mu_t(p)$. 

First, consider the case where $m \leq n$. Then we have
\[
s_{t,i} = (1 - \Delta t) \cdot s_{\lceil t \rceil, i} \;.
\]
Substituting $t=m/n$, $\lceil t \rceil=1$, and $\Delta t = \lceil t \rceil - t = 1 - m/n$ gives
\[
s_{t,i} = \frac{m}{n} \cdot (m-1) (1-p_i)^{m-2} = \Theta(1) \cdot \frac{m^2}{n} \;.
\]
Now consider the case where $n < m \leq \Theta(1) \cdot \frac{n}{\epsilon^2}$. 
Note that the case where $n < m < 2n$ follows from (i) the fact that $s_{t,i}$ for fractional $t$ linearly interpolates between the value of $s_{t',i}$ the nearest two integral values of $t'$ and (ii) the analyses of the cases where $m \leq n$ and $2n \leq m \leq \Theta(1) \frac{n}{\epsilon^2}$. 
Thus, all we have left to do is prove the case where $2n \leq m \leq \Theta(1) \cdot \frac{n}{\epsilon^2}$.

Since $s_{t,i}$ is a convex combination of $s_{\lceil t \rceil,i}$ and $s_{\lfloor t \rfloor,i}$, 
it suffices to bound from below these quantities for $t=m/n$. 
Both of these tasks can be accomplished simultaneously by bounding from below the quantity 
$s_{t=m/n+\gamma,i}$ for arbitrary $\gamma \in [-1,1]$. 

We do this as follows: Let $t=m/n+\gamma$.
Using Stirling's approximation, we will show that for any $\gamma\in [-1,1]$, we get:
\[
s_{t,i}\geq \Theta(1)\cdot \sqrt{mn} \;.
\] 
Note that Stirling's approximation is tight up to constant factors 
as long as the number we are taking the factorial of is not zero. 
Note that $m-2 \geq 1$, $t -1 \geq 1$, and $m-t-1 \geq m/2 - 2 \geq 1$. 
Thus, if we apply Stirling's approximation to the factorials in the definition 
of the binomial coefficient and substitute $t=m/n + \gamma$, we obtain 
the following approximation, which is tight up to constant factors:
\begin{align*}
\binom{m-2}{t-1}&=\Theta(1)\cdot \frac{\sqrt{m-2}}{\sqrt{m - m/n - 1-\gamma}\sqrt{m/n - 1+\gamma}}\frac{(m-2)^{m-2}}{(m - m/n - 1-\gamma)^{m - m/n - 1-\gamma}(m/n-1+\gamma)^{m/n-1+\gamma}} \\
&=\Theta(1)\cdot \sqrt{\frac{(m - m/n - 1-\gamma) (m/n - 1+\gamma)}{(m - 2)^3}} \cdot \frac{(m-2)^m  }{(m - m/n -1-\gamma)^{m - m/n-\gamma} (m/n+\gamma)^{m/n+\gamma}} \\
&=\Theta(1)\cdot \frac{1}{\sqrt{mn}} \frac{m^m}{(m - m/n -1-\gamma)^{m - m/n-\gamma} (m/n+\gamma)^{m/n+\gamma}}
\end{align*}

Using this approximation, we get: 
\begin{align*}
s_{t,i}&=(m-1) \binom{m-2}{t-1} p_i^{t-1} (1-p_i)^{m-t-1} \\
&= \Theta(1) \cdot \sqrt{\frac{m}{n}} \cdot \frac{m^m p_i^{m/n - 1+\gamma} (1 - p_i)^{m - m/n - 1-\gamma} }{(m - m/n -1-\gamma)^{m - m/n-\gamma} (m/n-1+\gamma)^{m/n+\gamma}} \\
&= \Theta(1) \cdot \frac{1}{p_i}\sqrt{\frac{m}{n}} \cdot \frac{m^{m/n+\gamma} p_i^{m/n+\gamma} (1 - p_i)^{m - m/n - 1-\gamma} }{(1 - \frac{1}{n} - \frac{1}{m})^{m - m/n-\gamma} (m/n)^{m/n+\gamma}} \\
&= \Theta(1) \cdot \frac{1}{p_i}\sqrt{\frac{m}{n}} \cdot \frac{(np_i)^{m/n+\gamma} (1 - p_i)^{m - m/n - 1 -\gamma} }{(1 - \frac{1}{n})^{m - m/n-\gamma}} \\
\end{align*}
By substituting $p_i=\frac{1\pm \varepsilon}{n}$, we get: 
\begin{align*}
s_{t,i}&= \Theta(1) \cdot \sqrt{mn} \cdot \frac{(1 \pm \epsilon)^{m/n+\gamma-1} (1 - \frac{1 \pm \epsilon}{n})^{m - m/n-\gamma} }{(1 - \frac{1}{n})^{m - m/n-\gamma}} \\
&= \Theta(1) \cdot \sqrt{mn} \cdot \frac{(1 \pm \epsilon)^{m/n} (1 - \frac{1 \pm \epsilon}{n})^{m - m/n} }{(1 - \frac{1}{n})^{m - m/n}} \\
&= \Theta(1) \cdot \sqrt{mn} \cdot (1 \pm \epsilon)^{m/n} \left(1 \mp \frac{\epsilon}{n-1}\right)^{m - m/n} \\
&\geq \Theta(1) \cdot \sqrt{mn} \cdot (1 \pm \epsilon)^{m/n} \left(1 \mp \epsilon\right)^{\frac{m}{n-1} - \frac{m}{n(n-1)}} \\
&= \Theta(1) \cdot \sqrt{mn} \cdot (1 \pm \epsilon)^{m/n} \left(1 \mp \epsilon\right)^{m/n} \\
&= \Theta(1) \cdot \sqrt{mn} \cdot (1 - \epsilon^2)^{m/n} \\
&= \Theta(1) \cdot \sqrt{mn} \cdot e^{-\Theta(1) \cdot \epsilon^2(m/n)} \\
&\geq \Theta(1) \cdot \sqrt{mn} \; (\mbox{since }m<\Theta(1)\cdot\frac{n}{\varepsilon^2}) \;.
\end{align*}
This completes the proof of Lemma~\ref{lem:hessian-approx}.
\end{proof}

\subsection{Concentration of Test Statistic: Proof of Theorem~\ref{thm:main-unif}} \label{ssec:concen}
Let the $m$ samples be $Y_1, \dotsc, Y_m \in [n]$, and let $X_i$, $i \in [n]$, be
the number of $j \in [m]$ for which $Y_j = i$.  
Let $S$ be our empirical total variation test statistic, 
$S = \frac{1}{2} \sum_{i=1}^n \abs{\frac{X_i}{m} - \frac{1}{n}}$.
We prove the theorem in two parts, one when $m \geq n$, and one when
$m \leq n$.

We will require a ``Bernstein'' form of the standard bounded
differences (McDiarmid) inequality:
\begin{lemma}[Bernstein version of McDiarmid's inequality~\cite{ying2004mcdiarmid}]\label{lem:mcdiarmid}
Let {$Y_1,\dots, Y_m$} be independent random variables taking values in the set $\mathcal{Y}$. 
Let $f: \mathcal{Y}^m \rightarrow \R$ be a function of {$y_1, \ldots ,y_m$} so that for every {$j \in [m]$}
and $y_1, \ldots, y_m, y_j^\prime \in \mathcal{Y}$, we have that:
\[
 \vert f(y_1,\ldots, y_j, \ldots, y_m) - f(y_1, \ldots, y_j^\prime, \ldots, y_m) \vert \leq B \;.
\]
Then, we have:
\begin{align}
\Pr\left[ f(Y_1, \ldots, Y_m)-\E[f] \geq z \right] \leq \exp\left(\frac{-2z^2}{{m} B^2}\right) \;. \label{eq:2}
\end{align}
If in addition, for each $j \in [m]$ and $y_1, \dotsc, y_{j-1}, y_{j+1}, \dotsc, y_m$ we have that
\[
  \Var_{Y_j} \left[ f(y_1, \ldots, Y_j, \dots, y_m) \right] \leq \sigma_j^2 \;,
\]
then we have
\begin{align}
\Pr\left[ f(Y_1, \ldots, Y_m)-\E[f] \geq z \right] \leq \exp\left(\frac{-z^2}{2\sum_{j=1}^m \sigma_j^2 + 2Bz/3}\right) \;. \label{eq:3}
\end{align}
\end{lemma}

\subsubsection{Case I: $m \geq n$}
 Since the $Y_j$'s are independent and $S$ is $\frac{1}{m}$-Lipschitz
in them, the first form of McDiarmid's inequality implies that 
\[
  \Pr[S - \E[S] \geq z] \leq \exp(-2mz^2) \;,
\]
and similarly, by applying it to $-S$, we have 
$\Pr[S - \E[S] \leq -z] \leq \exp(-2mz^2)$.

Let $R$ be the right-hand side of the Equation in
Lemma~\ref{lem:expectation-gap}, so $\mu(p) - \mu(U_n) \geq R$ in the
soundness case.  Since we threshold the tester at $t = \mu(U_n) + R/2$, we
find in both the completeness and soundness cases that the success probability will be at least
\[
    1 - \exp(-mR^2/2) \;,
\]
and hence we just need to show
\begin{align}
mR^2/2 \geq \log(1/\delta) \;. \label{eq:1}
\end{align}
Since we are in the regime that $m \geq n$, there are two possible cases
in Lemma~\ref{lem:expectation-gap}. 

For $n \leq m \leq n/\eps^2$, we need that
\[
 \frac{m}{2} \cdot \Theta(1) \cdot \eps^4m/n \geq \log(1/\delta)  
\]
or
\[
  m \geq \Theta(1) \cdot \frac{\sqrt{n\log(1/\delta)}}{\eps^2}  \;.   
\]

For $m \geq n/\eps^2$, we need that
\[
  \frac{m}{2} \cdot \Theta(1) \cdot \eps^2 \geq \log(1/\delta)
\]
or 
\[ 
m \geq  \Theta(1) \cdot \frac{\log(1/\delta)}{\eps^2} \;.
\]
The theorem's assumption on $m$ implies that both conditions hold, 
which completes the proof of Theorem~\ref{thm:main-unif} in this case.

\subsubsection{Case II: $m \leq n$}

To establish Theorem~\ref{thm:main-unif} for $m \leq n$, we will require 
the Bernstein form of McDiarmid's inequality (Equation (\ref{eq:3}) in Lemma~\ref{lem:mcdiarmid}).

To apply this form of Lemma~\ref{lem:mcdiarmid}, 
it suffices to compute $B$ and $\sigma_j$ for our
test statistic as a function of the $Y_j$'s.  Note that for $m \leq n$,
$\abs{\frac{X_i}{m} - \frac{1}{n}}$ is equal to
$\frac{X_i}{m} - \frac{1}{n}$ whenever $X_i \neq 0$.  In particular,
this implies
\begin{align*}
S &= \frac12 \sum_{i=1}^n \left\vert \frac{X_i}{m}-\frac1n \right\vert \\ 
&=\frac{1}{2} \sum_{i=1}^n \left\{ \left(\frac{X_i}{m}-\frac1n\right) + \frac{2}{n}\cdot \mathbf{1}_{X_i = 0} \right\} \\
&=\frac1n \cdot \vert \{ i:X_i=0 \} \vert \;.
\end{align*}
Hence, the value of the parameter $B$ for our test statistic is $1/n$, 
since each $Y_j$ will affect the number of nonzero $X_i$'s by at most $1$. 
In particular, the function value as $Y_j$ varies and the other $Y_{j'}$'s 
are kept fixed can be written as the sum of a deterministic quantity plus 
$(1/n) \cdot b$, where $b$ is a Bernoulli random variable that is $1$ 
if sample $Y_j$ collides with another sample $Y_{j'}$ and $0$ otherwise. 
Thus, the variance of $S$ as $Y_j$ varies and the other $Y_{j'}$'s are kept fixed 
is given by $\var[(1/n) \cdot b]$. This variance is $(1/n^2) \cdot r(1-r)$, 
where $r$ is the probability that $Y_j$ collides with another $Y_{j'}$.

\new{
By Fact \ref{cor:stochastic domination}, we have that probability that value of the statistic falls below the threshold (i.e type II error) is maximized for some distribution that has probability mass values in the set $\{\frac{1+\varepsilon^\prime}{n},\frac1n,\frac{1-\varepsilon^\prime}{n}\}$ for some $\varepsilon^\prime \geq  \dtv(p,U_n)/2$. This allows us to consider only this worst case. Therefore, we have that $r\leq m(1+\varepsilon^\prime)/n\leq 2m/n$ and} the variance of $S$ as $Y_j$ varies and the other $Y_{j'}$'s are kept fixed is at most
\[
\frac{1}{n^2} \cdot  r(1-r) \leq r/n^2 \leq \new{2}m/n^3 =: \sigma_j^2 \;.
\]
\new{in both the completeness and soundness cases.}

 Applying Equation (\ref{eq:3}) of Lemma~\ref{lem:mcdiarmid} we find
\[
\Pr[\abs{S - \E[S]} \geq z] \leq {2} \exp\left(\frac{-z^2}{\new{4} \cdot m^2/n^3 + (2/3) \cdot z/n}\right) \;.
\]
By Lemma~\ref{lem:expectation-gap}, in the soundness
case we have expectation gap $\mu(p) - \mu(U_n) \geq {R:=} C\eps^2m^2/n^2$
for some constant $C < 1$.
Substituting $z = R/2$ in the above concentration inequality 
yields that our tester will be correct with
probability $1-\delta$ as long as
\[
  m \geq \Theta(1) \cdot \frac{1}{\eps^2}\sqrt{n\log(2/\delta)},
\]
for an appropriately chosen constant, which is true by assumption.
This completes the proof of Theorem~\ref{thm:main-unif}.

\section{Conclusions and Future Work} \label{sec:future}
{In this paper, we gave the first uniformity tester that is sample-optimal, up to constant factors, as a function of the confidence parameter. 
Our tester is remarkably simple and our novel analysis may be useful in other related settings. By using a known reduction of identity to uniformity, 
we also obtain the first sample-optimal identity tester in the same setting.

Our result is a step towards understanding the behavior of distribution testing problems in the high-confidence setting. We view
this direction as one of fundamental theoretical and important practical interest. A number of interesting open problems remain.
Perhaps the most appealing one is to design a {\em general technique} (see, e.g.,~\cite{DK16}) that yields sample-optimal testers in the high confidence regime
for a wide range of properties.
From the practical standpoint, it would be interesting to perform a detailed experimental evaluation of the various algorithms 
(see, e.g.,~\cite{HM13-it, BalakrishnanW17}).
}

\bibliographystyle{alpha}
\bibliography{allrefs}

\newpage

\appendix

\section*{Appendix}

\section{Computation of the Expectation in Completeness Case} \label{app:exp}
Our statistic can be written as: $S=\sum_{i=1}^{n} \max\{ X_i-\frac{m}{n},0 \}$.
Therefore, by linearity of expectation, we get: 
\[
\E[S]=\sum_{i=1}^n \E\left[\max\{ X_i-\frac{m}{n},0 \}\right]=n\cdot \E\left[\max\{ X_1-\frac{m}{n},0 \}\right]  \;. 
\]
So, all we need to do is to compute: $\E\left[\max\{ X_i-\frac{m}{n},0 \}\right]$ for a single value of $i$ in the completeness case. 

Note that $X_i \sim \mathrm{Bin}(m,\frac1n)$ and that the above expectation can be written as:
\begin{equation}\label{expectation}
\E\left[\max\{ X_i-\frac{m}{n},0 \}\right]=\sum_{k=\lceil \frac{m}{n} \rceil}^m  \Pr[X_i=k](k-\frac{m}{n}) \;,
\end{equation}
where 
$\Pr[X_i=k]=\frac{(1-\frac1n)^{m-k}}{n^k}$. This is a sum of $O(m)$ terms each of which can be computed in constant time, giving an $O(m)$ runtime overall.

\section{Non-Existence of Indefinite Closed-Form for Components of Expectation}\label{app:no-closed-form}

In this appendix, we formalize and prove our assertion from Section~\ref{ssec:exp-gap} that 
the function $\mu(p)-\mu(U_n)$ is a linear combination of sums each of which has no indefinite closed form.

Recall Equation~\eqref{eq:expectation-formula} which says that the expectation is a linear combination of sums with summands of the form 
$\binom{m}{k}q^k (1-q)^{m-k} (k-t)$ for various values of $q$---where the values of $q$ are themselves different variables 
that any closed form would need to depend on (in addition to the other variables). 
A sum is said to have an indefinite closed form if, when the upper and lower limits of the sum are replaced with new variables, 
the resulting sum has a closed form valid for all values of all variables.

By closed form, we mean a closed form as defined in \cite[Definition 8.1.1]{PWZ97} which, as far as we are aware, 
is the main formal sense in which the phrase is used in combinatorics. 
This definition of closed form says that a function can be written as a sum of a constant number of rational functions, 
where the numerator and denominator in each is a linear combination of a constant number of products of exponentials, 
factorials, and constant degree polynomials. An example of such a function is $\frac{1}{\binom{n}{k}} + 7 \cdot \frac{2^k k! + 2k^5}{3^k}+k$.

To prove that a sum with summands $\binom{m}{k}q^k (1-q)^{m-k} (k-t)$ has no indefinite closed form---where $m, k, q, t$, 
and the limits of the sum are the variables that the closed form would need to be a function of---one can run Gosper's algorithm on this summand---
with $k$ as the index of summation---and observe that it returns that there is no indefinite closed form solution in the sense we have described \cite[Theorem 5.6.3]{PWZ97}, \cite{G78,PS95}.

\section{Omitted Proofs from Section \ref{sec:uniform}}

\subsection{Proof of Lemma \ref{lem:expectation-gap}}
\label{appendix:exp-gap lemma}
Given Lemmas~\ref{lem:hessian-exact} and \ref{lem:hessian-approx} from section \ref{ssec:exp-gap}, we are ready to prove the desired expectation gap.

\begin{proof}[{\bf Proof of Lemma~\ref{lem:expectation-gap}:}]
We start by reducing the soundness case to a much simpler setting.
To do this, we use the following fact, established in Section~\ref{sec:sd}:

\new{
\begin{fact}\label{cor:stochastic domination}
Let $S(D)$ be the random variable taking the value of our test statistic when the samples come from the distribution $D$. For any distribution $p$ on $[n]$, there exists a distribution $p^\prime$ supported on $[n]$
whose probability mass values are in the set $\{\frac{1+\varepsilon^\prime}{n},\frac1n,\frac{1-\varepsilon^\prime}{n}\}$ for some $\varepsilon^\prime \geq  \dtv(p,U_n)/2$, 
with at most one element having mass $\frac1n$, and such that the statistic $S(p)$ stochastically dominates $S(p^\prime)$. In particular, we have that {$\mu(p^\prime) \leq \mu(p)$}.
\end{fact}    
}

Fact~\ref{cor:stochastic domination} is proven in Appendix~\ref{sec:sd}. By Fact~\ref{cor:stochastic domination}, there is a distribution $p^\prime$ 
that satisfies the conditions of Lemma \ref{lem:hessian-approx}, 
has total variation distance $\Theta(\eps)$ to the uniform distribution, 
and {$\mu(p^\prime) \leq \mu(p)$}.
Therefore, it suffices to prove a lower bound on the expectation gap between the completeness and soundness cases for distributions $p$ of this form. 

Note that all probability distributions on the line from $p$ to $U_n$ are also of this form for different (no larger) values of $\eps$.  
Thus, Lemma~\ref{lem:hessian-approx} gives a lower bound on the diagonal entries of the Hessian at all points on this line. 
Since the Hessian is diagonal, this also bounds from below the minimum eigenvalue of the Hessian on this line. 
Therefore, by this and Equation~(\ref{eq:sc-lb}), we obtain the first two cases of this lemma, 
as well as the third case for $\frac{n}{\eps^2} \leq m \leq 4 \cdot \frac{n}{\eps^2}$.

The final case of this lemma for $4 \cdot \frac{n}{\eps^2} \leq m$ follows immediately 
from  the folklore fact that if one takes at least this many samples, the empirical distribution approximates 
the true distribution with expected $\ell_1$ error at most $\eps/2$. For completeness, we give a proof. We have
\newcommand\numberthis{\addtocounter{equation}{1}\tag{\theequation}}
\begin{align*}
\E[\|X / m - p\|_1] &= \sum_i \E[|X_i / m - p_i|] \leq \sum_i \sqrt{\var[X_i / m - p_i]} \\
&\leq \sum_i \sqrt{mp_i / m^2} \leq \sum_i \sqrt{m(1/n) / m^2} \numberthis \label{step:concavity} \\
&= \sqrt{n/m} \leq \eps/2 \;,
\end{align*}
where Equation~(\ref{step:concavity}) follows 
from the fact that the sum is a symmetric concave function of $p$, so it is maximized by setting all the $p_i$'s to be equal.
\end{proof}

\section{Matching Information-Theoretic Lower Bound} \label{sec:lb}

In this section, we prove our matching sample complexity lower bound. 
Namely, we prove:

\begin{theorem} \label{thm:lb}
Any algorithm that distinguishes with probability at least $1-\delta$ 
the uniform distribution on $[n]$ from any distribution that is $\eps$-far from uniform,
in total variation distance, requires at least
\[ \Omega\left(\left(\sqrt{n \log(1/\delta)}+\log(1/\delta)\right)/\varepsilon^2\right)
\]
samples.
\end{theorem}

Theorem~\ref{thm:lb} will immediately follow from separate sample complexity 
lower bounds of $\Omega(\log(1/\delta)/\eps^2)$ and 
$\Omega(\sqrt{n \log(1/\delta)}/\eps^2)$ that we will prove.
We start with a simple sample complexity lower bound of $\Omega(\log(1/\delta)/\eps^2)$:

\begin{lemma}\label{lem:lb-log}
For all $n, \eps$, and $\delta$, any $(\eps, \delta)$ uniformity tester 
requires $\Omega(\log(1/\delta)/\eps^2)$ samples.
\end{lemma}

\begin{proof}
If $n$ is odd, set the last probability to $1/n$, subtract $1$ from $n$, and invoke the following lower bound instance on the remaining elements. If $n$ is even, do the following. Consider the distribution $p$ which has probability $p_i=\frac{1+\varepsilon}{n}$ for each element $1\leq i\leq \frac{n}{2}$ 
and  $p_i=\frac{1-\varepsilon}{n}$ for each element $\frac{n}{2}\leq i \leq n$. 
Clearly, $\dtv(p, U_n) = \eps$.
Note that the probability that a sample comes from the first half of the domain is $\frac{1+\varepsilon}{2}$ and the probability that it comes from the second half of the domain is $\frac{1-\varepsilon}{2}$. Therefore, distinguishing $p$ from $U_n$ is equivalent to distinguishing between a fair coin 
and an $\eps$-biased coin. It is well-known (see, e.g., Chapter~2 of~\cite{Bar-Yossef02}) 
that this task requires $m=\Omega(\log(1/\delta)/\eps^2)$ samples.
\end{proof}

The rest of this section is devoted to the proof of the following lemma, which gives our desired lower bound:

\begin{lemma} \label{lem:lb-sqlog}
For all $n, \eps$, and $\delta$, any $(\eps, \delta)$ uniformity tester 
requires at least $\Omega(\sqrt{n \log(1/\delta)}/\eps^2)$ samples. 
\end{lemma}

To prove Lemma~\ref{lem:lb-sqlog}, we will construct two indistinguishable families of \emph{pseudo-distributions}. 
A pseudo-distribution $w$ is a non-negative measure, i.e., it is similar to a probability distribution except 
that the ``probabilities'' may sum to something other than $1$. We will require that our pseudo-distributions 
always sum to a quantity within a constant factor of $1$. A pair of pseudo-distribution families 
is said to be \emph{$\delta$-indistinguishable using $m$ samples} if no tester exists that can, 
for every pair of pseudo-distributions $w, w'$---one from each of the two families---distinguish the product distributions 
$\operatorname*{\bigotimes} \Poi(m w_i)$ versus $\operatorname*{\bigotimes} \Poi(m w'_i)$ with failure probability at most $\delta$.

This technique is fairly standard and has been used in \cite{WY16,VV14,DK16} to establish lower bounds for distribution testing problems. 
The benefit of the method is that it is much easier to show that pseudo-distributions are indistinguishable, 
as opposed to working with ordinary distributions. Moreover, lower bounds proven using pseudo-distributions 
imply lower bounds on the original distribution testing problem.

We will require the following lemma, whose proof is implicit in the analyses of \cite{WY16,VV14,DK16}:

\begin{lemma}\label{lem:poi-lb}
Let $\mathcal{P}'$ be a property of distributions. 
We extend $\mathcal{P}'$ to the unique property $\mathcal{P}$ of pseudo-distributions 
which agrees with $\mathcal{P}'$ on true distributions and is preserved under rescaling. 
Suppose we have two families $\mathcal{F}_1, \mathcal{F}_2$ of pseudo-distributions with the following properties:
\begin{enumerate}
\item All pseudo-distributions in $\mathcal{F}_1$ have property $\mathcal{P}$
and all those in $\mathcal{F}_2$ are $\eps$-far in total variation distance from any pseudo-distribution that has the property.
\item $\mathcal{F}_1$ and $\mathcal{F}_2$ are $\delta$-indistinguishable using $m$ samples.
\item Every pseudo-distribution in each family has $\ell_1$-norm within the interval $[\frac{1}{c_1},c_2]$, for some constants $c_1, c_2>1$. 
\end{enumerate}
Then there exist two families $\widetilde{\mathcal{F}}_1,\widetilde{\mathcal{F}}_2$ of probability distributions with the following properties:
\begin{enumerate}
\item All distributions in $\widetilde{\mathcal{F}}_1$ have property $\mathcal{P}$ 
and all those in $\widetilde{\mathcal{F}}_2$ are $\frac{\eps}{c_2}$-far in total variation distance from any distribution that has the property.
\item Any tester that can distinguish $\widetilde{\mathcal{F}}_1$ and $\widetilde{\mathcal{F}}_2$ has worst-case 
error probability $\geq \delta-2^{-c m}$, for some constant $c>0$ or requires $\Theta(1) \cdot m$ samples.
\end{enumerate}
\end{lemma}

In our case, the property of distributions $\mathcal{P}'$ is simply being the uniform distribution. 
The families of pseudo-distributions we will use for our lower bound are the family $\mathcal{F}_1$ 
that only contains the uniform distribution and the family $\mathcal{F}_2$ of all pseudo-distributions of the form 
$w_i = \frac{1 \pm \eps}{n}$ such that $|1-\sum_i w_i| \leq \eps/2$. Note that this constraint on the sum of the $w_i$'s 
ensures the first and last conditions needed to invoke Lemma~\ref{lem:poi-lb}\footnote{If we did not have this constraint, 
the first condition would not be satisfied (because e.g., a $w$ such that $w_i = (1+\eps)/n$ for all $i$ is not $\eps$-far 
from being proportional to the all-ones vector.}.

Furthermore, by Lemma~\ref{lem:lb-log}, the required number of samples $m$ 
satisfies $m \geq \Omega(\log(1/\delta))$. 
Ignoring constant factors,  we may assume that $m \geq c' \log(1/\delta)$ for any constant $c'>0$. 
In particular, by selecting $c'$ appropriately, we can guarantee that $2^{-cm} \leq \delta/3$, where $c$ 
is the constant in the last statement of Lemma~\ref{lem:poi-lb}.
Thus, the error probability guaranteed by Lemma~\ref{lem:poi-lb} 
for distinguishing the true distribution families is at least $(2/3) \delta$.

Thus, all that remains is to show that $\mathcal{F}_1$ and $\mathcal{F}_2$ 
are $\delta$-indistinguishable using $m$ samples. In order to show these families are indistinguishable, 
we show that it is impossible to distinguish whether the product distribution 
$\operatorname*{\bigotimes} \Poi(m w_i)$ has $w$ uniform or $w$ generated 
according to the following random process: we pick each $w_i$ independently by setting 
$w_i=\frac{1+\eps}{n}$ or $w_i=\frac{1-\eps}{n}$ each with probability $1/2$.

A distribution generated by this process has a small probability of not being in $\mathcal{F}_2$. 
Specifically, this happens iff it fails to satisfy the constraint on having a sum within $1 \pm \eps/2$. 
However, by an application of the Chernoff bound, it follows that this happens with probability at mots $\leq 2^{-\Theta(1) \cdot n}$. 
Since the bound in Lemma~\ref{lem:lb-log} is larger (up to constant factors) than the lower bound we presently wish to prove in the case that $\delta/3 > 2^{-\Theta(1) \cdot n}$, we will assume $\delta \leq 2^{-\Theta(1) \cdot n}$; 
in which case, the following lemma implies that we can still invoke Lemma~\ref{lem:poi-lb}, 
where the probability of not being in $\mathcal{F}_2$ is absorbed into our overall indistinguishability probability, 
and we get a final indistinguishability probability of at least $(2/3) \delta - \delta/3 = \delta/3$.

The following lemma is implicit in \cite{WY16,VV14,DK16}:

\begin{lemma}\label{lem:dist-to-poi-lb}
Let $\mathcal{P}$ be a property of pseudo-distributions. 
Suppose we have two families $\mathcal{F}_1,\mathcal{F}_2$ 
of pseudo-distributions and two distributions $\mathcal{D}_1,\mathcal{D}_2$ 
on $\mathcal{F}_1$ and $\mathcal{F}_2$ respectively with the following properties:
\begin{enumerate}
\item With probability each at least $1-\delta_1$, a distribution output by $\mathcal{D}_1$ is in $\mathcal{F}_1$ and a distribution output by $\mathcal{D}_2$ is in $\mathcal{F}_2$.
\item If we generate $w$ according to $\mathcal{D}_1$ or $\mathcal{D}_2$, then any algorithm for determining which family $w$ 
came from given access to $\operatorname*{\bigotimes} \Poi(m w_i)$ has worst case error probability at least $\delta_2$.
\end{enumerate}
Then $\mathcal{F}_1$ and $\mathcal{F}_2$ are $(\delta_2 - \delta_1)$-indistinguishable using $m$ samples.
\end{lemma}

Thus, we now simply need to show that $\mathcal{D}_1$ and $\mathcal{D}_2$ are hard to distinguish. 
Let $X_i,X_i^\prime$ be the random variables equal to the number of times the element $i$ is sampled 
in the completeness and soundness cases respectively. We will require a technical lemma that will be used to 
bound the Hellinger distance between any pair of corresponding coordinates in the completeness and soundness cases.
By $ \frac{1}{2}\Poi(\lambda_1) + \frac{1}{2}\Poi(\lambda_2)$, we denote a uniform mixture of the corresponding distributions.
We have

\begin{fact}[Lemma 7 of~\cite{VV14}] \label{fact:hellinger-Poisson-mixtures}
  For any $\lambda > 0, \eps < 1$ we have
  \[
    H^2\left( \Poi(\lambda), \frac{1}{2}\Poi((1+\eps)\lambda) + \frac{1}{2}\Poi((1-\eps)\lambda) \right) \leq C\lambda^2 \eps^4.
  \]
  for some constant $C$.
\end{fact}

We also require the following Lemma which gives a tighter relationship between Hellinger distance when the distance is close to $1$.

\begin{lemma}\label{lem:hellinger-tv-bound}
Any distributions with Hellinger-squared distance $H^2(p,q) \leq 1-\delta$ have total variation distance at most $\dtv(p,q) \leq 1-\delta^2/2$.
\end{lemma}

The more standard inequality between these quantities only gives $\dtv(p,q) \leq \sqrt{2(1-\delta)}$ which is worse than the trivial bound of $1$ when $\delta$ is small.

\begin{proof}
Let $H^2(p,q) \leq 1-\delta$. Let $a_i:=\min(p_i,q_i)$ and $b_i=\max(p_i,q_i)$. Then we have
\begin{gather*}
1 - \delta \geq H^2(p,q) = \frac{1}{2} \sum_i (\sqrt{p_i} - \sqrt{q_i})^2 = 1 - \sum_i \sqrt{p_i} \sqrt{q_i} \\
\delta^2 \leq \left( \sum_i \sqrt{p_i} \sqrt{q_i} \right)^2 = \left( \sum_i \sqrt{a_i} \sqrt{b_i} \right)^2.
\end{gather*}

Applying Cauchy-Schwarz then yields
\[
\delta^2 \leq \left( \sum_i a_i \right) \left( \sum_i b_i \right) \leq (1-\dtv(p,q)) \cdot 2.
\]
Thus,
\[
\dtv(p,q) \leq 1-\delta^2/2.
\]
\end{proof}

\noindent We are now ready to prove Lemma~\ref{lem:lb-sqlog}.

\begin{proof}[{\bf Proof of Lemma~\ref{lem:lb-sqlog}:}]
As follows from the discussion preceding Fact~\ref{fact:hellinger-Poisson-mixtures}, it suffices 
to show that $\mathcal{D}_1$ and $\mathcal{D}_2$ are hard to distinguish. 
We will use Fact~\ref{fact:hellinger-Poisson-mixtures} to show that the Hellinger 
distance between the overall distributions is small, which implies their total variation distance is small, 
and hence that they cannot be distinguished with probability better than $\delta$.

Recall that each of the $n$ coordinates of the vectors output by these distributions 
is distributed according to $\Poi(m/n)$ for $\mathcal{D}_1$ vs. a uniform mixture 
of $\Poi((1\pm\eps)m/n)$ for $\mathcal{D}_2$.  A single coordinate then has
\[
 H^2(X_1, X_1') \leq C(m^2/n^2)\eps^4 \;,
\]
so the collection of all coordinates has
\[
  H^2(X, X') \leq 1 - (1 - C(m^2/n^2)\eps^4)^n \leq 1 - e^{-C(m^2/n)\eps^4}
\]
so by Lemma~\ref{lem:hellinger-tv-bound},
\[
  \dtv(X, X') \leq 1 - e^{-2C(m^2/n)\eps^4}/2.
\]
For this latter quantity to be at least $1-\delta$, we need
\[
  m =\Omega\left( (1/\eps^2) \cdot \sqrt{n \log (1/\delta)} \right)
\]
samples, as desired. The result then follows from
Lemma~\ref{lem:dist-to-poi-lb} and Lemma~\ref{lem:poi-lb}.  This
completes the proof.
\end{proof}

\section{Stochastic Domination for Statistics of the Histogram} \label{sec:sd}
In this section, we consider the set of statistics which 
are symmetric convex functions of the histogram 
(i.e., the number of times each domain element is sampled) of an arbitrary random variable $Y$.
We start with the following definition:
\begin{definition}\label{def:majorization}
Let $p=(p_1,\dots,p_n),q=(q_1,\dots,q_n)$ be probability distributions
and $p^{\downarrow},q^{\downarrow}$ denote the vectors with the same values as $p$ and $q$ respectively, 
but sorted in non-increasing order. We say that $p$ \emph{majorizes} $q$ (denoted by $p\succ q$) if 
\begin{equation}\label{eq:condition}
\forall k: \sum_{i=1}^k p^{\downarrow}_i \geq \sum_{i=1}^k q^{\downarrow}_i \;.
\end{equation}
\end{definition}

\noindent The following theorem from \cite{arnold12} gives an equivalent definition:
\begin{theorem}{\cite{arnold12}}\label{thm:majorization doubly stochastic}
Let $p=(p_1,\dots,p_n),q=(q_1,\dots,q_n)$ be any pair of probability distributions. Then, $p\succ q$ if and only if there exists a doubly stochastic matrix $A$ such that $q=Ap$. 
\end{theorem}

\noindent \textbf{Remark:}
It is shown in \cite{arnold12} that multiplying the distribution $p$ by a doubly stochastic matrix is equivalent 
to performing a series of so called ``Robin hood operations'' and permutations of elements. 
Robin hood operations are operations in which probability mass in transferred 
from heavier to lighter elements. For more details, the reader is referred to \cite{arnold12,MOA79}.

\medskip

\noindent {Note that Definition \ref{def:majorization} defines a partial order 
over the set of probability distributions. We will see that the uniform distribution is a minimal element 
for this partial order, which directly follows as a special case of the 
following lemma.}

\begin{lemma}\label{lem:averaging-implies-majorization}
Let $p$ be a probability distribution over $[n]$ and $S\subseteq [n]$. 
Let $q$ be the distribution which is identical to $p$ on $[n]\setminus S$, 
and for every $i \in S$ we have $q_i=\frac{{p}(S)}{\vert S\vert}$, 
where $\vert S\vert$ denotes the cardinality of $S$. 
Then, we have that $p\succ q$.    
\end{lemma}

\begin{proof}
Let $A=(a_{ij})$ be the doubly stochastic matrix $A=(a_{ij})$ with entries: 
\[
a_{ij}= \begin{cases}
1 & \textrm{ if } i=j\not\in S\\ 
\frac{1}{|S|} & \textrm{ if } i \in S \wedge j \in S \\
0   &  \textrm{otherwise}
\end{cases} \;.
\]
Observe that $q=Ap$. Therefore, Theorem \ref{thm:majorization doubly stochastic} implies that $p \succ q$.
\end{proof}

{In the rest of this section, we use the following standard terminology: 
We say that a real random variable $A$ stochastically dominates a real random variable $B$ 
if for all $x \in \R$ it holds $\Pr[A > x] \geq \Pr[B > x]$.}
We now state the main result of this section (see Section~\ref{ssec:dom} for the proof):

\begin{lemma}\label{lm:general domination}
Let $f: \R^n \to \R$ be a symmetric convex function and $p$ be a distribution over $[n]$. 
Suppose that we draw $m$ samples from $p$, and let $X_i$ denote the number of times we sample element $i$. 
Let $g(p)$ be the random variable $f(X_1,X_2,\ldots,X_n)$. 
Then, for any distribution $q$ over $[n]$ such that $p\succ q$, 
we have that $g(p)$ stochastically dominates $g(q)$.
 \end{lemma}

\medskip

As a simple consequence of the above, we obtain the following:

\begin{fact}\label{cor:mean-domination}
Let $p$ be a distribution on $[n]$ and $S \subseteq [n]$. 
Let $p^\prime$ be the distribution which is identical to $p$ in $[n]\setminus S$ and the probabilities in $S$ are averaged 
(i.e., $p_i^\prime=\frac{p(S)}{\vert S\vert}$). Then, we have that $\mu(p^\prime)\leq \mu(p)$, where $\mu(\cdot)$ denotes the expectation of our statistic as defined in Section~\ref{sec:uniform}. In particular, $\mu(U_n) \leq \mu(p)$ for all $p$.
\end{fact}

\begin{proof}
Recall that our statistic applies a symmetric convex function $f$ to the histogram of the sampled distribution.
Since $p^\prime$ is averaging the probability masses on a subset $S\subseteq [n]$, 
Lemma~\ref{lem:averaging-implies-majorization} gives us that $p\succ p^\prime$.
Therefore, by  Lemma \ref{lm:general domination} we conclude that $g(p)$ stochastically
dominates $g(p^\prime)$, which implies that:
$\mu(p)=\E[g(p)] \geq \E[g(p^\prime)]=\mu(p^\prime)$, as was to be shown.
\end{proof}

The following lemma shows that given an arbitrary distribution $p$ over $[n]$ that is $\varepsilon$-far from the uniform distribution $U_n$, 
if we average the heaviest $\lfloor n/2\rfloor$ elements and then the lightest $\lfloor n/2\rfloor$ elements, 
we will get a distribution that is  $\varepsilon^\prime > \varepsilon/2$-far from uniform. 
\begin{lemma}\label{lem:averaging-preserves-tv}
  Let $p$ be a probability distribution and
  $p^\prime$ be the distribution obtained from
  $p$ after averaging the
  $\lfloor\frac{n}{2}\rfloor$ heaviest and the
  $\lfloor\frac{n}{2}\rfloor$ lightest elements separately.  Then, the
  following holds:
\[
\frac{\Vert p-U_n\Vert_1}{2}\leq \Vert p^\prime-U_n\Vert_1\leq \Vert p-U_n\Vert_1 \;.
\] 
\end{lemma}

\medskip

\noindent We note that by doing the averaging as suggested by the above lemma, we obtain a distribution $p^\prime$ that is supported 
on the following set of three values: $\{\frac{1+\varepsilon^\prime}{n},\frac1n,\frac{1-\varepsilon^\prime}{n}\}$, 
for some $\frac{\varepsilon}{2}\leq \varepsilon^\prime\leq \varepsilon$. 
Hence, we can reduce the computation of the expectation gap for an arbitrary distribution $p$, 
to computing the gap for a distribution of this form. 
Fact~\ref{cor:stochastic domination} is an immediate corollary of
Lemmas~\ref{lem:averaging-implies-majorization},~\ref{lm:general domination}, and~\ref{lem:averaging-preserves-tv}.

\subsection{Proof of Lemma \ref{lm:general domination}} \label{ssec:dom}

To establish Lemma \ref{lm:general domination}, we are going to use the following intermediate lemmas: 
\begin{lemma}\label{yensen}
Let $f:\mathbb{R}^n\rightarrow \mathbb{R}$ be a symmetric convex function, 
and $a,b,c\in \mathbb{R}$ such that $0<a<b$ and $c>0$. Then, 
\[
f(a,b+c,x_3,\dots,x_n)\geq f(a+c,b,x_3,\dots,x_n) \;.
\]
\end{lemma}
\begin{proof}
Consider the set of convex functions $f^\prime_{x_3,\dots,x_n}:\mathbb{R}^2\rightarrow \mathbb{R}$ defined as:
 \[f^\prime_{x_3,\dots,x_n}(x_1,x_2)=f(x_1,x_2,x_3,\dots,x_n) \;.\]
 We will show that for every possible choice of $x_3,\dots,x_n$ it holds that:
 \[
  f^\prime_{x_3,\dots,x_n}(a,b+c)\geq f^\prime_{x_3,\dots,x_n}(a+c,b) \;.
 \]
Since $f$ is symmetric, so is $f^\prime$. Therefore, we have that $f^\prime_{x_3,\dots,x_n}(a,b+c)=f^\prime_{x_3,\dots,x_n}(b+c,a)$.
The $3$ points: $P_1=(a,b+c), P_2=(a+c,b), P_3=(b+c,a)$ are collinear 
since their coordinates satisfy the equation $x_1+x_2=a+b+c$.

We have that $P_2$ is between $P_1$ and $P_3$ since:
\[
\langle\vec{P_1P_2}, \vec{P_2P_3}\rangle= \langle(c,-c), (b-a,a-b) \rangle>0 \;.
\]
By applying Jensen's inequality, we get that    
\[f^\prime_{x_3,\dots,x_n}(a+c,b)\leq \frac{f^\prime_{x_3,\dots,x_n}(a,b+c)+f^\prime_{x_3,\dots,x_n}(b+c,a)}{2}=f^\prime_{x_3,\dots,x_n}(a,b+c) \;. \]
as desired. 
\end{proof}

The stochastic domination between the two statistics is established in the following lemma:
\begin{lemma}\label{lm:domination}
Let $f:\mathbb{R}^n\rightarrow \mathbb{R}$ be a symmetric convex function, 
$p$ be a distribution over $[n]$, and $a, b \in [n]$ be such that $p_a<p_b$. 
Also, let $q$ be the distribution which is identical to $p$ on $[n]\setminus \{a,b\}$, and for which: 

\begin{equation} \label{2-element averaging repeat}
\begin{pmatrix}
q_a \\
q_b 
\end{pmatrix}=
\begin{pmatrix}
w & 1-w \\
1-w & w
\end{pmatrix}
\begin{pmatrix}
p_a  \\
p_b 
\end{pmatrix} \;,
\end{equation}
where $w\in [\frac12,1]$. Suppose we take $m$ samples from $p$ and let $X_i$ denote the number of times we sample element $i$. 
Let $g(p)$ be the random variable $f(X_1,X_2,\dots,X_n)$. Then, $g(p)$ stochastically dominates $g(q)$.    

\end{lemma}

\begin{proof}
To prove stochastic domination between $g(p)$ and $g(q)$, we are going to define a coupling under which it is always true that $g(p)$ takes a larger value than $g(q)$.

Initially, we define an auxiliary coupling between $p$ and $q$ as follows: 
To get a sample from $q$, we first sample from $p$ and use it as our sample, 
unless the output is element ``$b$'', in which case we output ``$a$'' with probability 
$\frac{(1-w)(p_b-p_a)}{p_b}$ and ``$b$'' otherwise\footnote{Note that this coupling does not fix the value of $g(q)$ 
given a fixed value for $g(p)$, and is defined for convenience. 
We still have to show stochastic domination for the coupled random variables using a second coupling.}. 

Suppose now that we draw $m$ samples from $p$, which we also convert to samples from $q$ using the above rule. 

In relation to this coupling we define the following random variables:
\begin{itemize}
\item $X_{low}$: The number of times element ``$a$'' is sampled. 
\item $X_{high}$: The number of times element ``$b$'' is sampled and \emph{is not} swapped for element ``$a$'' in $q$.  
\item $X_{mid}$: The number of times element ``$b$'' is sampled and \emph{is} swapped for element ``$a$'' in $q$. 
\end{itemize}
From the above, we have that:
\[
X_a=X_{low} ,\quad X_b=X_{high}+X_{mid},\quad X_a^\prime=X_{low}+X_{mid},\quad X_b^\prime=X_{high} \;,
\]
where $X_i^\prime$ is the number of times element $i$ is sampled in $q$. 

We want to show that $g(p)$ stochastically dominates $g(q)$.
That is, we want to show that\footnote{To simplify notation, we pick $a=1$ and $b=2$ without loss of generality.}:
\[
 \forall t: \Pr[f(X_{low},X_{high}+X_{mid},X_3,\dots,X_n)\geq t]\geq \Pr[f(X_{low}+X_{mid},X_{high},X_3,\dots,X_n)\geq t] \;.
\]
We now condition on the events 
$E_{\{y,z\}}:\{X_{low},X_{high}\}=\{y,z\}$ and $E_{c,x_3,\dots,x_n}:X_{mid}=c\wedge X_3=x_3\wedge \dots \wedge X_n=x_n\}$, 
where $y\leq z$ without loss of generality. 
Let $B=E_{\{y,z\}}\wedge E_{c,x_3,\dots,x_n}$.

We have that:
\begin{align*}
&\Pr[f(X_{low},X_{high}+X_{mid},X_3,\dots,X_n)\geq t]\\
&\quad=\sum_{y\leq z} \Pr[f(X_{low},X_{high}+X_{mid},X_3,\dots,X_n)\geq t\mid B]\Pr[B] \;.
\end{align*}
So, it suffices to show that for every $y,z,c,x_3,\dots,x_n,t:$ 
\begin{align}
&\phantom{= {}}\Pr[f(X_{low},X_{high}+X_{mid},X_3,\dots,X_n)\geq t\mid B]\notag\\
&\geq \Pr[f(X_{low}+X_{mid},X_{high},X_3,\dots,X_n)\geq t\mid B]  \;.\label{eq:goaldom}
\end{align}

At this point, we have conditioned on everything except which of
$X_{low}$ and $X_{high}$ is $y$ and which is $z$.  That is, after
conditioning on the event $B=E_{\{y,z\}}\wedge E_{c,x_3,\dots,x_n}$,
we have that:
\[
\{f(X_{low}+X_{mid},X_{high},X_3,\dots,X_n),f(X_{low},X_{high}+X_{mid},X_3,\dots,X_n) \}= \{u,w\} \;,
\]
where $u=f(y+c,z,x_3,\dots,x_n),w=f(y,z+c,x_3,\dots,x_n)$.  Since by
assumption $y\leq z$, we have by Lemma \ref{yensen} that $u\leq w$.  Then
\eqref{eq:goaldom} holds trivially as an equality for $t\leq u $ and
for $ t>w$.  For the remaining values of $t$, it is equivalent to:
\[
\Pr[f(X_{low},X_{high}+X_{mid},X_3,\dots,X_n)=w\mid B]\geq \Pr[f(X_{low}+X_{mid},X_{high},X_3,\dots,X_n)=w\mid B],
\]
and hence to
\[
  \Pr[X_{low}=y,X_{high}=z \mid B]\geq\Pr[X_{low}=z,X_{high}=y \mid B].
\]
Now, this is also equivalent to a version with less restricted conditioning,
\[
  \Pr[X_{low}=y,X_{high}=z \mid E_{c,x_3,\dotsc,x_n}]\geq\Pr[X_{low}=z,X_{high}=y \mid E_{c,x_3,\dotsc,x_n}],
\]
because neither event occurs in the added regime where $E_{\{y,z\}}$
is false.  But if we rethink how our samples were drawn, we find that
this is equivalent to showing that
\[
  p_a^yq_b^z \geq p_a^zq_b^y.
\]
This holds since $q_b^{z-y}>p_a^{z-y}$, concluding the proof.
\end{proof}

\begin{proof}[{\bf Proof of Lemma \ref{lm:general domination}:}]
Since $p\succ q$, we have by Theorem~\ref{thm:majorization doubly stochastic} (and the remark that follows it)
that $q$ can be constructed from $p$ by repeated applications of  \eqref{2-element averaging repeat}. 
Therefore, Lemma \ref{lm:domination} and the fact that stochastic domination is transitive imply that $g(p)$ stochastically dominates $g(q)$.
\end{proof}

\subsection{Proof of Lemma \ref{lem:averaging-preserves-tv}}
Recall that $p^\downarrow_i$ denotes the vector $p$ with entries rearranged in non-increasing order.
Suppose that at least $n/2$ elements have at least $1/n$ probability mass\footnote{This is without loss of generality, since we can use essentially the same argument in the other case.}.
Therefore, if $p$ is not the uniform distribution, we have
\[
\sum_{k=1}^{\lfloor n/2\rfloor}p^\downarrow_k=\frac{1 + \varepsilon^\prime}{2} > \frac12 \;,
\] 
for some $\varepsilon^\prime >0$.

Thus, we have that 
\begin{align*}
p_{k}^{\prime\downarrow} = \begin{cases}
      \frac{1+\varepsilon^\prime}{n} & \text{for } k\leq\frac{n}{2}\\
      \frac{1-\varepsilon^\prime}{n} & \text{for }  k>\frac{n}{2}
   \end{cases}
   \end{align*}
   when $n$ is even, and 
   \begin{align*}
p_{k}^{\prime\downarrow} = \begin{cases}
      \frac{1+\varepsilon^\prime}{n} & \text{for } k\leq\frac{n-1}{2}\\
       \frac1n & \text{for } k=\frac{n+1}{2}  \\
      \frac{1-\varepsilon^\prime}{n} &  \text{for } k>\frac{n+1}{2}
   \end{cases}
\end{align*}
when $n$ is odd.

Moreover, since we are just averaging, we have that 
$
\sum_{k=1}^{\lfloor n/2\rfloor}p_{k}^{\downarrow}=\sum_{k=1}^{\lfloor n/2\rfloor}p_{k}^{\prime\downarrow}=\frac{1 + \varepsilon^\prime}{2} \;.
$
Since we have assumed 
that the majority of elements has mass at least $\frac1n$, we know that the total variation distance is given by:
\[
\dtv(p,U_n)= \sum_{i:p_i> 1/n} (p_i -1/n) \leq 2 \sum_{k=1}^{\lfloor n/2\rfloor} (p^\downarrow_i - 1/n) 
= 2 \sum_{k=1}^{\lfloor n/2\rfloor} (p^{\prime\downarrow}_i - 1/n)
\leq 2 \sum_{i:p_i^\prime>1/n} (p^{\prime\downarrow}_i - 1/n)= 2\dtv(p^\prime,U_n) \;.
\]
Thus, $\dtv(p^\prime,U_n)\geq (1/2) \dtv(p,U_n)$ or $\Vert p-U_n \Vert_1 /2 \leq \Vert p^\prime-U_n \Vert_1$, as desired.

\end{document}